\documentclass[english,12pt]{article}

\usepackage{amsmath}
\usepackage{times}
\usepackage{bm}
\usepackage{natbib}
\usepackage{fullpage}
\usepackage{comment}
\usepackage[plain,noend]{algorithm2e}
\usepackage{algorithmic}
\usepackage{algorithm2e}
\usepackage{amssymb}
\usepackage{mathrsfs}
\usepackage{placeins}
\usepackage{bbm}
\usepackage{float} 
\usepackage[colorlinks=true,linkcolor=black,citecolor=blue]{hyperref}
\usepackage{xcolor}
\usepackage{amsthm}

\usepackage{cleveref}
\crefname{assumption}{Assumption}{Assumptions}
\crefname{remark}{Remark}{Remarks}
\crefname{proposition}{Proposition}{Propositions}
\crefname{theorem}{Theorem}{Theorems}
\crefname{section}{Section}{Section}
\crefname{lemma}{Lemma}{Lemma}
\crefname{algorithm}{Algorithm}{Algorithms}
\crefname{example}{Example}{Examples}
\crefname{figure}{Figure}{Figure}
\crefname{appendix}{Appendix}{Appendix}
\crefname{equation}{equation}{equation}

\newtheorem{theorem}{Theorem}

\newtheorem{proposition}[theorem]{Proposition} 
\newtheorem{assumption}{Assumption} 
\newtheorem{remark}{Remark} 
\newtheorem{lemma}[theorem]{Lemma}
\usepackage{mathtools,xparse}
\usepackage{amsfonts}

\def\trans{{T}}

\DeclarePairedDelimiter{\abs}{\lvert}{\rvert}
\DeclarePairedDelimiter{\norm}{\lVert}{\rVert}

\def\EE{E}
\def\RR{{\mathbb R}}

\def\PP{{\mathbb{P}}}

\def\ndata{{n}}

\def\Up{U^{0}}

\def\strata{\mathcal{S}}

\def\iweight{\omega}

\def\target{{\pi}}

\def\x{{\bm x}}

\def\pos{{\xi}}

\def\vel{{\theta}}

\def\Lc{{\mathcal{L}}}

\def\dd{{\rm d}}

\newcommand{\N}{\mathrm{Normal}}

\newcommand{\U}{\mathrm{Uniform}}

\newcommand{\OO}{\mathcal{O}}
\renewcommand{\RR}{\mathbb{R}}

 \DeclareMathOperator*{\argmin}{\arg\!\min}

 \def\Natt{N_{\mathrm{attempts}}}

\def\llh{\bullet}

\def\effgain{\mathcal{E}}

\usepackage{enumerate}

\allowdisplaybreaks
\begin{document}

\title{Efficient posterior sampling for high-dimensional imbalanced logistic regression}

\author{Deborshee Sen$^{1}$\footnote{The two authors contributed equally to this paper.} $^,$\footnote{Corresponding author.} \\ {ds430@duke.edu} 
\and 
Matthias Sachs$^{2,}$\footnotemark[1]\\ {msachs@math.duke.edu}  
\and 
Jianfeng Lu$^{2}$ \\
{jianfeng@math.duke.edu}
\and
David Dunson$^{1,2}$ \\ {dunson@duke.edu} }

\date{$^1$Department of Statistical Science, Duke University \\
$^2$Department of Mathematics, Duke University}

\maketitle 

\begin{abstract}
\noindent 
Classification with high-dimensional data is of wide-spread interest and often involves dealing with imbalanced data. Bayesian classification approaches are hampered by the fact that current Markov chain Monte Carlo algorithms for posterior computation are inefficient as the number of predictors $p$ or the number of subjects to classify $n$ gets large due to worsening computational time per step and mixing rates.
One strategy is to use a gradient-based sampler to improve mixing while using data sub-samples to reduce per-step computational complexity.  However, usual sub-sampling breaks down when applied to imbalanced data. Instead, we generalize piece-wise deterministic Markov chain Monte Carlo algorithms to include importance-weighted and mini-batch sub-sampling. These maintain the correct stationary distribution with arbitrarily small sub-samples, and substantially outperform current competitors. We provide theoretical support and illustrate gains in simulated and real data applications.   
\end{abstract}


\section{Introduction}

In developing algorithms for large datasets, much of the focus has been on optimization algorithms that produce a point estimate with no characterization of uncertainty. This motivates scalable Bayesian algorithms.  As variational methods and related analytic approximations lack theoretical support and can be inaccurate, this article focuses on posterior sampling algorithms.  

One such class of methods is divide-and-conquer Markov chain Monte Carlo, which divides data into chunks, runs Markov chain Monte Carlo independently for each chunk, and then combines samples  \citep{li2017simple, srivastava2018scalable}. However, combining samples inevitably leads to some bias, and accuracy theorems require sample sizes to increase within each subset.

An alternative strategy uses sub-samples to approximate transition probabilities and reduce bottlenecks in calculating likelihoods and gradients \citep{welling2011bayesian}. Such approaches typically rely on uniform random sub-samples, which can be highly inefficient, as noted in an increasing frequentist literature on biased sub-sampling \citep{fithian2014local, ting2018optimal}. 
The Bayesian literature has largely overlooked the use of biased sub-sampling in efficient algorithm design, though recent coreset approaches address a related problem \citep{huggins2016coresets}.
A problem with sub-sampling Markov chain Monte Carlo is that it is almost impossible to preserve the correct invariant distribution. While there has been work on quantifying the error \citep{johndrow2015optimal, alquier2016noisy}, it is usually difficult to do so in practice. 
The pseudo-marginal approach of \cite{andrieu2009pseudo} offers a potential solution, but it is generally impossible to obtain the required unbiased estimators of likelihoods using sub-samples \citep{jacob2015nonnegative}. 

A promising recent direction has been using non-reversible samplers with sub-sampling within the framework of piece-wise deterministic Markov processes \citep{bouchard2018bouncy, fearnhead2018piecewise, bierkens2019zig}. These approaches use the gradient of the log-likelihood, which can be replaced by a sub-sample-based unbiased estimator, so that the exactly correct invariant distribution is maintained. This article focuses on improving the efficiency of such approaches by using non-uniform sub-sampling motivated concretely by logistic regression.

\section{Logistic regression with sparse imbalanced data} \label{sec.logistic_PDMP} 

\subsection{Model} \label{sec.model.logistic}

We focus on the logistic regression model
\begin{align} \label{eq.model}
P(y = 1 \mid x, \pos)
& = 
\{ 1+ \exp ( -x^T \pos ) \}^{-1},
\end{align}
where $y \in \{0,1\}$ is the response, $x = (x_1, \ldots, x_p) \in \RR^p$ are predictors, and $\pos = (\pos_1, \ldots, \pos_p) \in \RR^p$ are coefficients for the predictors. Consider data $(y^1, x^1), \ldots, (y^n, x^n)$ from model \eqref{eq.model}, where $x^j = (x_1^j, \ldots, x_p^j) \in \RR^p$.
For a prior distribution $p_0(\pos)$ on $\pos$, the posterior distribution is $\pi(\pos) = p_0(\pos) \prod_{j=1}^n [\exp \{ y^j (x^j)^T \pos \}]/[1 + \exp \{ (x^j)^T \pos \} ] = \exp\{ -U(\pos) \}$,
where $(x^j)^T \pos = \sum_{i=1}^p x_i^j \pos_i \in \RR$ and $U(\pos)$ denotes the potential function.  
A popular algorithm for sampling from $\pi(\xi)$ is P\'olya-Gamma data augmentation \citep{polson2013bayesian}. However, this algorithm performs poorly if there is a large imbalance in class labels $y$ \citep{johndrow2019mcmc}. Similar issues arise when the $x_i^j$s are imbalanced. Logistic regression is routinely used in broad fields and such imbalance issues are extremely common, leading to a clear need for more efficient algorithms.
While standard Metropolis-Hastings algorithms not relying on data augmentation can perform well despite imbalanced data in settings where both $n$ and $p$ are small \citep{johndrow2019mcmc}, issues arise in scaling to large datasets due to increasing computational time per step and slow mixing.

\subsection{The zig-zag process} \label{sec.ZZ}

The zig-zag process \citep{bierkens2019zig} is a type of piece-wise deterministic Markov process which is particularly useful for logistic regression. 
The zig-zag process is a continuous-time stochastic process $\{\pos(t),\vel(t)\}_{t \geq 0}$ on the augmented space $\RR^p \times \{ -1,1 \}^p$, where $\pos(t)$ may be understood as the position and $\vel(t)$ the velocity of the process at time $t$. Under fairly mild conditions, the zig-zag process is ergodic with respect to the product measure $\widetilde{\pi}(\dd \pos, \dd \vel) = \pi(\dd \pos) \otimes \upsilon(\dd \vel)$, where $\upsilon(\dd \vel)$ is the uniform measure on $\{-1, 1\}^p$. 
In other words, $\EE_{(\eta,\beta) \sim \widetilde{\pi}} \{\varphi(\eta,\beta)\} = \lim_{T \rightarrow \infty} T^{-1} \int_{0}^{T} \varphi\{\pos(t),\vel(t)\}\, \dd t$ holds almost surely for any $\widetilde{\pi}$-integrable function $\varphi$.

For a starting point $\pos$ and velocity $\vel$, the zig-zag process evolves deterministically as 
\begin{equation} \label{eq:det:evol}
\pos(t) 
 = 
\pos + \theta t, ~~ \vel(t) = \vel.
\end{equation}
At a random time $\tau$, a bouncing event flips the sign of one component of the velocity $\theta$. The process then evolves as \cref{eq:det:evol} with the new velocity until the next change in velocity. 
The time $\tau$ is the first arrival time of $p$ independent Poisson processes with intensity functions $m_1(t), \dots, m_p(t)$, that is, $\tau = \tau_{i_{0}}$ with $i_0 = \argmin_{i \in \{1, \dots, p\}} \{ \tau_i \}$.  The sign flip applies $F_{i_{0}} : \{ -1,1 \}^p \rightarrow \{ -1,1 \}^p$ to $\vel(t)$, with $\{F_i(\theta)\}_k = \theta_k$ if $k \neq i$ and $-\theta_k$ if $k=i$.
The intensity functions are of the form $m_{i}(t) = \lambda_{i}\{\pos(t),\vel(t)\}$, where $\lambda_{i}$ is a rate function. A sufficient condition for the zig-zag process to preserve $\widetilde{\pi}$ as its invariant distribution is the existence of non-negative functions $\gamma_1, \dots, \gamma_p$, such that $\lambda_{i}(\pos,\vel)  = \gamma_{i}(\pos) + \left \{ \vel_{i}\partial_{i}U(\pos) \right \}^+$ $(i=1,\dots,p)$; here $(x)^+ = \max\{0,x\}$ denotes the positive part of $x \in \RR$. The $\gamma_i$s are known as refreshment rates.

If $\Lambda_{i}(t) = \int_{0}^{t}m_{i}(s) \,\dd s$ has a simple closed form, the arrival times $\tau_{i}$ can be sampled as $\tau_{i} = - \log \Lambda_{i}^{-1}(u)$ for $u\sim \U\{(0,1)\}$. Otherwise, $\tau_{i}$ are obtained via Poisson thinning  \citep{lewis1979simulation}. 
Assume that we have continuous functions $M_i(t) : \RR_+ \rightarrow \RR_+ $ such that $m_i(t) = \lambda_{i}(\pos + t \vel,\vel) \leq M_i(t)$ $(i = 1, \dots, p; ~ t \geq 0)$;
here $M_i(t)$ are upper computational bounds.
Let $\widetilde{\tau}_1, \ldots, \widetilde{\tau}_p$ denote the first arrival times of non-homogeneous Poisson processes with rates $M_1(t), \dots, M_p(t)$, respectively, and let $i_0 = \argmin_{ i \in \{ 1, \ldots, p \} } \{ \widetilde{\tau}_i \}$. A zig-zag process with intensity $m_i(t)=\lambda_i\{\pos(t),\vel(t)\}$ is still obtained if $\pos(t)$ is evolved according to \cref{eq:det:evol} for time $\widetilde{\tau}_{i_{0}}$ instead of $\tau_{i_{0}}$ and the sign of $\vel_{i_{0}}$ is flipped at time $\widetilde{\tau}_{i_{0}}$ with probability $m_{i_0}(\widetilde{\tau}_{i_{0}}) / M_{i_0}(\widetilde{\tau}_{i_{0}})$.

The sub-sampling approach of \cite{bierkens2019zig} uses uniform sub-sampling of a single data point to obtain an unbiased estimate of 
the $i$-th partial derivative of the potential function 
$U(\pos) = U^{0}(\pos) + U^{\llh} (\pos)$, 
where $U^{0} (\pos) = -\log p_{0}(\pos)$ is from the prior and $U^{\llh}(\pos) = \sum_{j=1}^n U^{j}(\pos)$ with $U^{j}(\pos) = - \log  p(y^j \mid x^j, \pos)$ $(j= 1, \dots, n)$ is from the likelihood. Their sub-sampling algorithm preserves the correct stationary distribution. The authors consider estimates $\widehat{\partial}_{i} U (\pos,J)$ such that $\EE_{J \sim \U [\{1,\dots, n\}]} \{ \widehat{\partial}_{i} U (\pos,J)\} = \partial_i U(\pos)$, where $J$ indexes the sampled data point.
This is used to construct a stochastic rate function as $\widehat{m}_{i}^{J}(t) = \{ \vel_{i} \widehat{ \partial}_{i}U (\pos+t\theta,J) \}^{+}$.
By using upper bounds satisfying $\max_{j \in \{ 1, \dots, n\} }\widehat{m}_{i}^{j}(t)  \leq M_{i}(t)$ for all $t \geq 0$,
the rate functions $m_{i}(t)$ can be replaced by stochastic versions $\widehat{m}_{i}^{J}(t)$, with $J$ being resampled at every un-thinned event. 

In addition, control variates can be used to reduce the variance of the estimate $\widehat{\partial}_{i} U (\pos,J)$. 
This can lead to dramatic increases in sampling efficiency when the posterior is concentrated around a reference point \citep{bierkens2019zig}.
In this article, we use isotropic Gaussian priors and focus on situations where either $p$ is large relative to $n$, or the covariates and/or responses are imbalanced. In such situations, the posterior is not sufficiently concentrated around a reference point for control variates to perform efficiently. We demonstrate this numerically in \cref{sec.scaling_CV} for imbalanced responses, and the Appendix contains similar experiments for sparse covariates. For this reason and due to space constraints, we focus our discussion on sub-sampling techniques that do not rely on control variates; the techniques developed can be combined with the use of control variates, as detailed in the Appendix.

\section{Improved sub-sampling} \label{sec.improved_sub-sampling} 

\subsection{General framework} \label{sec:uniform:ss}

We introduce a generalized version of the zig-zag sampler.
Our motivation is to (a) increase the sampling efficiency, and (b) simplify the construction of upper bounds. We achieve (b) by letting the Poisson process that determines  bouncing times in component $i\in \{1,\dots,p\}$ be a super-positioning (that is, the sum) 
 of two independent Poisson processes with state-dependent bouncing rates 
$\lambda_i^0 (\pos,\vel) = \{ \vel_i \partial_i  \Up (\pos) \}^{+}$ and 
$\lambda_i^{\llh}(\pos,\vel) = \{  \vel_i \partial_i  U^{\llh}(\pos) \}^+ + \gamma_{i}(\pos)$, respectively. Such a construction allows Poisson thinning of each process separately, which decouples the problem of constructing suitable upper bounds for the prior and the likelihood, respectively. We achieve (a) through general forms of the estimator $\widehat{\partial}_{i}U$ in the Poisson thinning step obtained through non-uniform sub-sampling. 

The resulting algorithm is presented as \cref{alg.zigzag.generalized}, where $ m_{i}^{0}(t) =  \{\vel_{i} \partial_{i}\Up(\pos+t\vel) \}^{+}$ and
$m^{\llh}_{i}(t, a) = \{ \vel_{i} \widehat{ \partial}_{i} U^{\llh}(\pos+t\vel,a) \}^+$, assuming that $\widehat{ \partial}_{i} U^{\llh}(\pos,a)$, $a \sim \mu_{i}$, is an unbiased estimator of  $\partial_{i}U$, and $M^{\llh}_{i}(t)$ is such that for all $a$ and $t \geq 0$, $m^{\llh}_{i}(t, a)  \leq M^{\llh}_{i}(t)$. To keep the presentation simple, we do not explicitly include the Poisson thinning step for the prior in the algorithm. The state-dependent bouncing rate of the resulting zig-zag process is $\lambda_i (\pos, \theta) = \lambda_i^0 (\pos, \theta) + \lambda_i^{\llh}(\pos, \theta)$,
with  $\lambda^{\llh}_{i}(\pos,\theta)$ having the explicit form $\lambda_i^{\llh}(\pos, \theta)  =  
\EE_{a \sim \mu_{i}} [ \{ \vel_{i} \widehat{ \partial}_{i} U^{\llh}(\pos,a) \}^{+} ]$.
General results on piece-wise deterministic Markov processes imply that such a zig-zag process preserves the target measure $\widetilde{\target}$ \citep{fearnhead2018piecewise}; we nonetheless provide a proof in \cref{sec:lemma:proof} for the sake of a self-contained presentation.
\begin{algorithm}
\caption{Zig-zag algorithm with generalized sub-sampling} 
\label{alg.zigzag.generalized} 

\textbf{Input:} Starting point $\pos^{(0)} \in \RR^p$, initial velocity $\theta^{(0)} \in \{-1,1 \}^p$, maximum number of bouncing attempts $\Natt$.
\begin{algorithmic}[1] 

\STATE 
Set $t^{(0)}=0$.

\FOR{$k = 1, \ldots, \Natt$}

\STATE 
Draw $\widetilde{\tau}_i, \widehat{\tau}_i ~ (i =1,\ldots, p)$ such that 
$ \PP( \widetilde{\tau}_i \geq t) = \exp \{ - \int_0^t M^{\llh}_i(s) \, \dd s  \}$ 
and 
$\PP( \widehat{\tau}_{i} \geq t) = \exp \{ - \int_0^t m^{0}_i(s) \, \dd s \}.$

\STATE 
Set $\tau =  \min \{\widetilde{\tau}_{i_{0}}, \widehat{\tau}_{j_{0}} \}$ where  $i_0 = \argmin_{i} \{ \widetilde{\tau}_i \}, ~j_{0}= \argmin_{i } \{ \widehat{\tau}_i \}$.

\STATE 
Evolve position: $\{t^{(k+1)}, \pos^{(k+1)}\} = \{t^{(k)} + \tau, \pos^{(k)} + \theta^{(k)}  \tau \}$. 

\STATE 
Draw $B \sim \mu_{i_{0}}, u \sim \U\{(0,1)\}$. 

\STATE 
{\bf if} $\widehat{\tau}_{j_{0}}< \widetilde{\tau}_{i_{0}}$ 
{\bf then} $\vel^{(k+1)} = F_{j_{0}}\{ \vel^{(k)}\}$

\STATE
{\bf else if} 
$u < m^{\llh}_{i_0}(\tau,B)/ M^{\llh}_{i_0}(\tau)$
{\bf then} 
$\vel^{(k+1)} = F_{i_{0}}\{ \vel^{(k)}\}$ 

\STATE
{\bf else}  
$\vel^{(k+1)} = \vel^{(k)}$.


%
\ENDFOR
\end{algorithmic}
\textbf{Output:} The path of a zig-zag process specified by skeleton points $\{(\pos^{(k)}, \vel^{(k)})\}_{k=0}^{\Natt}$ and bouncing times $\{t^{(k)}\}_{k=0}^{\Natt}$. 

\end{algorithm}

Although the focus of this article is on sampling from the Bayesian logistic regression problem presented in \cref{sec.model.logistic}, the approach can be readily applied to situations where the following assumption on the terms $U^{j}$ in the log-likelihood applies.
\begin{assumption} \label{as:uniform}
The partial derivatives of $U^{j}$ are bounded, that is, there exist constants $c_{i}^{j}>0 ~ (i=1, \dots, p; j =1,\dots, n)$, such that for all $\pos \in \RR^p$, $|\partial_i U^{j}(\pos)| \leq c_{i}^{j}$.
\end{assumption}

For the logistic regression problem considered, it is shown in \cite{bierkens2019zig} that \cref{as:uniform} is satisfied with 
\begin{align*}
c_i^j = |x_i^j| ~~ (j = 1, \dots, n; ~ i = 1, \dots, p).
\end{align*}
To keep things simple, we consider the prior to be $\N_p(\mathbf{0}, \sigma^2 I_p)$; we discuss other choices of priors in \cref{sec.different_priors}. Then we have that $\{ \vel_i \partial_i  \Up (\pos + \theta \, t) \}^+ \leq (|\pos_i|+t)/\sigma^2$.

In the sequel, we introduce alternative sub-sampling schemes and associated estimators and bounds as variants of the zig-zag sampler. These are designed to improve sampling efficiency by either (a) improving the mixing of the zig-zag process, or (b) reducing the computational cost per simulated unit time interval.
More specifically, we replace uniform sub-sampling with importance sub-sampling (\cref{sec:importance:ss}) to address (b), and allow general mini-batches instead of sub-samples of size one (\cref{sec:mb:ss}) to address (a); \cref{sec:strat:ss} contains a further extension to stratified sub-sampling which allows for further improvements of mixing.

\subsection{Improving bounds via importance sampling} \label{sec:importance:ss}

A generalization of the estimator obtained using uniform sub-sampling $\widehat{\partial}_{i}U^{\llh} (\pos,J) = n \partial_i U^J(\pos)$, $J \sim \U [ \{1, \dots, n\}]$, is to consider the index $J$ to be sampled from a non-uniform probability distribution $\nu_{i}$, defined by $\nu_{i}[\{j\}] = \iweight_{i}^{j} ~ (j=1,\dots,n)$, where $\iweight_i^1, \dots, \iweight_i^n > 0$ are weights satisfying $\sum_{j=1}^n \iweight_i^j=1$ $(i=1, \dots, p)$. It  follows that $\widehat{\partial}_i U^{\llh} (\pos,J) = 
( \iweight_{i}^{J} )^{-1} \partial_i U^{J}(\pos)$, 
$J \sim \nu_{i}$,
defines an unbiased estimator of $\partial_{i}U^{\llh}$. Moreover, $M^{\llh}_{i}(t) = \widetilde{c}_{i}(\iweight)$ with $\widetilde{c}_{i}(\iweight) = \max_{j\in\{1,\dots,n\}} c_{i}^{j} / \iweight_{i}^{j}$ defines an upper bound for rate function $m^{\llh}_{i}(t, J)$ under \cref{as:uniform}. 
The contribution $\lambda^{\llh}_{i}$ to the effective  bouncing rate is $\lambda^{\llh}_{i}(\pos,\theta) = \EE_{J \sim \nu_{i}} [\{ \theta_i \,  \widehat{\partial}_i U^{\llh} (\pos,J) \}^+] = n^{-1} \sum_{j=1}^n \{ \ndata \vel_i \partial_i U^{j}(\pos) \}^{+}$,
which is the same as that for uniform sub-sampling (which corresponds to $w_i^j = n^{-1}$ $(i=1,\dots,p; ~ j = 1, \dots, n)$). 

The magnitudes of upper bounds $M^{\llh}_{i}(t)$ can be minimized by choosing the weight vector $\iweight_{i} = (\iweight_{i}^1, \dots, \iweight_{i}^n)$ such that the constants $\widetilde{c}_{i}$ are minimized. This can be verified to be the case when $\iweight_{i}^{j} = c_{i}^{j}/ \overline{c}_{i}$, $(j=1,\dots, n)$ with $\overline{c}_{i} = \sum_{j=1}^n c_{i}^{j}$ so that $\widetilde{c}_{i}(\iweight) = \overline{c}_{i}$. 
This approach can be trivially generalized to allow for importance weights $\iweight_{i}^{j}=0$ in cases where the respective partial derivative vanishes, that is, $\partial_i U^{j}(\xi) \equiv 0 \Rightarrow c_{i}^{j}=0$, which is, for example, the case when the respective covariate $x_{i}^{j}$ in the considered logistic regression example is zero. For logistic regression, using optimal importance sub-sampling thus reduces the bounds from $n \max_{j \in \{1, \dots, n \} }|x_i^j|$ to $\sum_{j=1}^n |x_i^j|$ for the $i$-th dimension. This reduction is particularly significant when the $x_i^j$s are sparse and/or have outliers (see \cref{sec:scaling_bounds}).

\subsection{Improving mixing via mini-batches}
\label{sec:mb:ss}

In the context of piece-wise deterministic Markov processes, the motivation for using mini-batches of size larger than one is to reduce the effective refreshment rate, which can be expected to improve the mixing of the process if the refreshment rate is high \citep{andrieu2021hypocoercivity}.  
We consider a mini-batch $B = (J_{1},\dots,J_m) \sim \mu_{i}$ of random indices $J_{k} \in \{1, \dots, n\} ~ (k=1,\dots,m$), so that  $\widehat{\partial}_{i}U^{\llh}(\pos,B)$ is an unbiased estimator of  $\partial_{i}U^{\llh}(\pos)$. Entries of the mini-batch are typically sampled uniformly and independently from the data set. This yields unbiased estimators of the form $\widehat{\partial}_{i}U^{\llh}(\pos,B) = m^{-1} \sum_{k=1}^m  
n \partial_i U^{J_{k}}(\pos)$, $J_1, \dots, J_m \sim
\U[\{1, \dots, n\}]$. 

Since for any function $g: \{1, \dots, n\} \rightarrow  \RR$, $\max_{b \in \{1, \dots, n\}^m } m^{-1} \sum_{k=1}^m g(b_{k}) = \max_{b \in \{1, \dots, n\}} g(b)$, it follows that upper bounds for mini-batch size $m=1$ are also upper bounds for mini-batch sizes $m>1$. We can also let $\widehat{\partial}_{i}U^{\llh}(\pos,B) = m^{-1} \sum_{k=1}^m \{ ( \iweight_{i}^{J_k} )^{-1} \partial_i U^{J_{k}}(\pos) \}$, $J_1, \dots, J_m \sim \nu_{i}$, where $B = (J_1, \dots, J_m)$ and $\nu_{i}$ and $\iweight_{i}$ are as defined in \cref{sec:importance:ss}; by the same arguments, we conclude that the value of $\max_{B \in \{1, \dots, n\}^m}\widehat{\partial}_{i}U^{\llh}(\pos,B)$ does not depend on the size of the mini-batch $B$. 

If we consider mini-batches of size $m>1$, the effective bouncing rate of the zig-zag process when used with the estimators described above can be computed as $$\lambda_i^{\llh, \, (m)}(\pos,\vel)
= n^{-m} \sum_{(j_{1}, \dots, j_m)  \in \{1, \dots, n\}^m} \left \{ m^{-1} \sum_{k=1}^{m}\ndata \vel_i \partial_i U^{j_{k}}(\pos) \right \}^{+}.$$
The effective refreshment rate $\gamma_{i}^{\llh,  \, (m)}(\pos) = \lambda_i^{\llh, \, (m)}(\pos,\vel) - \{  \vel_i \partial_i  U^{\llh}(\pos) \}^+$ is reduced with increased mini-batch size, as stated in the following lemma.
\begin{lemma} \label{prop.mini-batch_size} 
For all $\pos\in \RR^p$, $\vel\in \{-1,1\}^p$, $m\geq 1$, we have
$ \gamma_{i}^{\llh, \, (m+1)}(\pos)\leq \gamma_{i}^{\llh, \, (m)}(\pos).$
\end{lemma}

\section{Synthetic data examples} \label{sec:synthetic:data}

\subsection{Scaling of computational efficiency} 
\label{sec:scaling_bounds}

We evaluate sampling efficiency using synthetic data generated by sampling the covariates $x_i^j$ from mixture distributions of the form $\nu_{\alpha}(\dd x) = (1- \alpha) \delta_0(\dd x) + \alpha \rho(x)\dd x$, where $\delta_0(\dd x)$ is a point mass at zero, $\rho$ is a smooth density, and $\alpha \in (0,1]$ determines the degree of sparsity. The responses $y_i$ are sampled from \cref{eq.model} with true $\pos = \pos_{\rm true} \in \RR^{p}$; we choose $p=5$ in this section. We further choose a non-informative prior by setting the prior variance to be $\sigma^2 = 10^{10}$. We repeat the data generation and sampling $50$ times. The expected gain in efficiency using importance sub-sampling instead of uniform sub-sampling is estimated as $50^{-1} \sum_{k=1}^{50} (T^{k}_{{\rm unif}}/T^{k}_{{\rm imp}})$, where $T^{k}_{{\rm unif}}$ and $T^{k}_{{\rm imp}}$ denote the total simulation time after $10^5$ attempted bounces of the zig-zag process in the $k$-th run using uniform and importance sub-sampling, respectively. 
\cref{sec:numerical_CV} contains further, similar experiments but using control variates.

We can expect the expected relative gain in efficiency to behave as
\begin{equation*}
\effgain ( \nu_{\alpha} )
:= 
\min_{i \in \{1, \dots, p\}}\EE \left ( \frac{n \max_{j \in \{1,\dots,n\}} \abs{x_{i}^{j}}}{\sum_{j=1}^{n} \abs{x_{i}^j}} \right )
\approx
\frac{ \EE ( \max_{j \in \{1,\dots,n\}} \abs{x_{i}^j} ) }{\EE ( n^{-1} \sum_{j=1}^{n} \abs{x_{i}^j} ) }.
\end{equation*}
We first plot the gain in efficiency for sparse covariates for $n=500$ and decreasing $\alpha$ in the left panel of \cref{fig.scaling_bounds}. 
Indeed, the behavior of the estimated relative gain in efficiency as approximately $\alpha^{-1}$ is as suggested by the first order Taylor expansion of  $\effgain ( \nu_{\alpha} )$ if $n$ is large. Similarly, the form of $\effgain(\nu_{\alpha})$ suggests that the expected relative gain in efficiency is unbounded as the number of observations $n$ increases if  $\rho$ has unbounded support. 
To this end, we choose dense covariates for increasing $n$. As the right panel of \cref{fig.scaling_bounds} shows, the relative gain in efficiency for $\rho = {\rm Laplace}(0,1)$ and $\N(0,1)$ are of order $\log n$ and 
$\log ( n/\log n ) $
as $n\rightarrow \infty$, respectively, which is what the first order Taylor approximation of $\effgain(\nu_{\alpha})$ suggests. 
Additionally, when $\rho = \mathrm{Student}\text{'}\mathrm{s ~ t}(3)$, we observe the efficiency gain to be of order $n^{1/3}$.

\begin{figure}[ht]
\centering 
\includegraphics[width=0.8\textwidth]{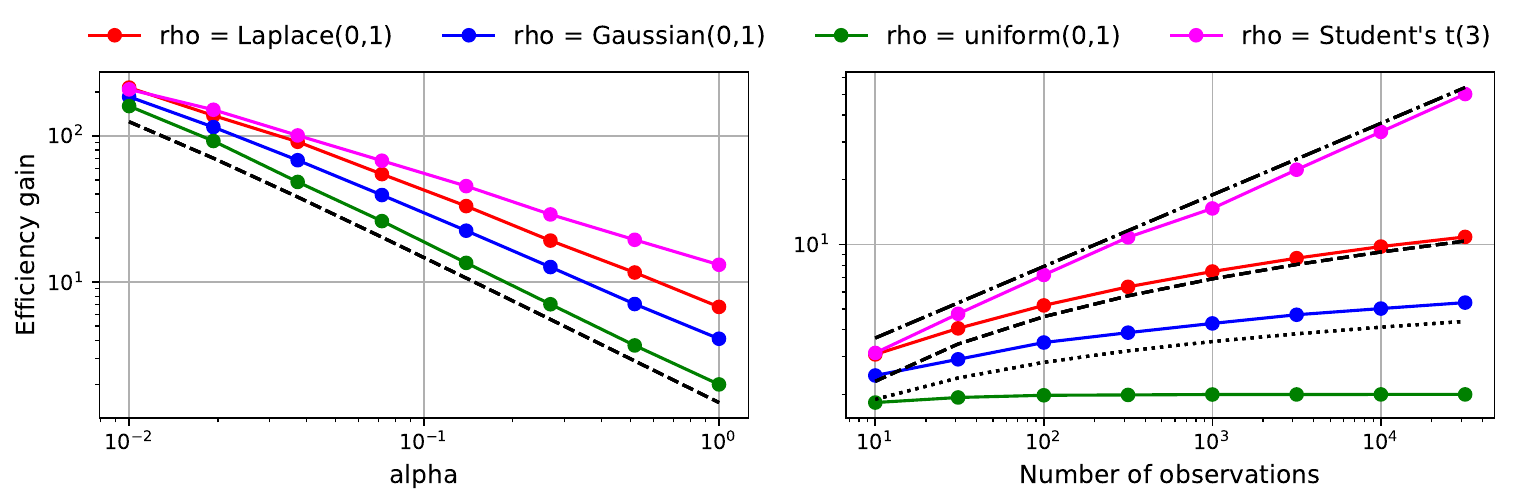}
\caption{Scaling of relative gain in efficiency.  Left panel: $n=500$; dashed black line is proportional to $\alpha^{-1}$.  Right panel: $\alpha=1$; dashed black line is proportional to $\log n$, dotted black line is proportional to $\log ( n/\log n ) (1+1/\log n)$, and dotted-dashed line is proportional to $n^{1/3}$.}
\label{fig.scaling_bounds}
\end{figure}

\subsection{Control variates for sparse data}  
\label{sec.scaling_CV}

As mentioned in \cref{sec.ZZ}, importance sub-sampling can be combined with the use of control variates. However, this fails to be efficient when the data are imbalanced and/or sparse. To demonstrate this, we generate covariates as described in \cref{sec:scaling_bounds} for $\alpha = 10^{-1}$ and $\rho = \mathrm{Laplace}(0,1)$, and generate responses independently of the covariates such that exactly $k$ of them are ones. We choose $p=10$ and $n=5 \times 10^3$, and choose the prior variance to be one. We plot the ratio of the mixing time of the slowest mixing component for importance sub-sampling with control variates divided by the same for importance sub-sampling without control variates as a function of $k$ in the left panel of \cref{fig.highdim_and_scaling}. 
As the responses become more imbalanced ($k$ decreases), the efficiency of using control variates decreases relative to not using them.

\subsection{High-dimensional sparse example} \label{sec.high_dim} 

We consider a challenging setting with number of dimensions $p=10^4$ and number of observations $\ndata=10^6$. The data are generated as in \cref{sec:scaling_bounds} with $\rho = \N(0,1)$ and $\alpha = 10^{-3}$. Traditional data augmentation and/or sub-sampling algorithms are either computationally very expensive or mix slowly in such a scenario. 
We choose the prior variance to be one and show auto-correlation function plots for uniform sub-sampling and importance sub-sampling in the center and right panels of \cref{fig.highdim_and_scaling}, respectively. This shows the necessity of using importance sub-sampling for the zig-zag sampler to be a feasible sampling method. From a practical point of view, adaptive pre-conditioning can further help, and this is described in \cref{sec.precond}.

\begin{figure}[ht]
\centering \includegraphics[width=\textwidth]{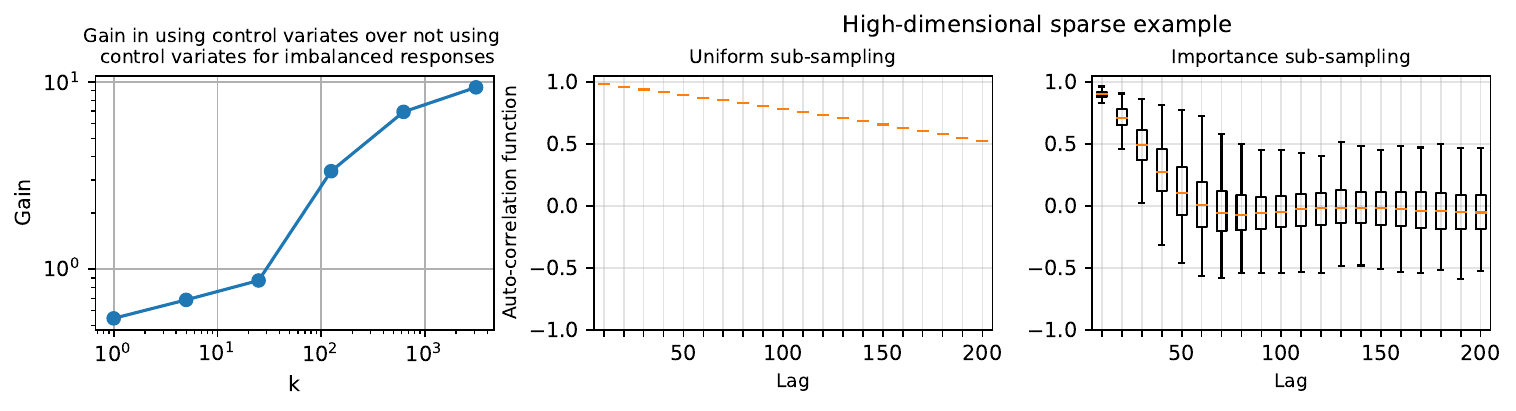}
\caption{Efficiency of using control variates (left panel), and high-dimensional sparse example (center and right panels).}
\label{fig.highdim_and_scaling}
\end{figure}

\section{Real data example} \label{sec.real_examples}

We consider an imbalanced dataset on cervical cancer \citep{fernandes2017transfer}, obtained from the machine learning repository \citep{Dua:2019}.
This has 858 observations with 34 predictors. The responses are whether an individual has cancer or not, with only 18 out of 858 individuals having cancer. The predictors include the number of sexual partners, hormonal contraceptives, et cetera, and more than half the predictors have approximately 80\% zeros. Fixing the number of bouncing attempts, the mixing times of the slowest mixing component for uniform sub-sampling and importance sub-sampling are $975.5$ and $193.2$, respectively. Furthermore, a stratification scheme described in \cref{sec:strat:ss} brings this down to $83.3$.

\section{Discussion} \label{sec.conclusion} 

Sub-sampling data is increasingly important in a variety of scenarios, particularly for big and high-dimensional data. Sub-sampling for traditional Markov chain Monte Carlo schemes can be tricky as the resulting chains induce an error in the invariant distribution that can be difficult to quantify. A promising recent class of algorithms known as piece-wise deterministic Markov processes allow sub-sampling without modifying the invariant measure. Non-uniform sub-sampling strategies can improve such algorithms as compared to uniform sub-sampling, especially for logistic regression with sparse covariate data; however, this can also be the case for other problems in which piece-wise deterministic Markov processes can be applied. 
After completion of this work, we became aware of a 2016 Oxford master's thesis by Nicholas Galbraith, where a method called informed sub-sampling is introduced, which is similar to the importance sub-sampling strategy presented here. While aspects of sparsity in the covariate data are not addressed there, the author makes similar observations regarding the usefulness of the approach in the setup of covariate data with outliers and/or distributed according to heavy-tailed distributions.

\section*{Acknowledgment}

This research was partially supported by grants from the United States National Science Foundation and National Institutes of Health. We are also grateful to the associate editor and two anonymous referees for helping us improve the paper.

\bibliographystyle{apalike}
\bibliography{references}

\newpage 
\appendix

\section{More details on the zig-zag process} 

\subsection{The zig-zag algorithm} \label{sec.zz_alg}

The zig-zag algorithm of \cite{bierkens2019zig} is presented in \cref{alg.zig-zag}.

\begin{algorithm}[ht]
\caption{Zig-zag algorithm with Poisson thinning \citep{bierkens2019zig}.} 
\label{alg.zig-zag} 

\textbf{Input:} Starting point $\pos^{(0)} \in \RR^p$, initial velocity $\theta^{(0)} \in \{-1,1 \}^p$, maximum number of bouncing attempts $\Natt$.

\begin{algorithmic}[1] 

\STATE 
Set $t^{(0)} = 0$.

\FOR{$k = 1, \ldots, \Natt$}

\STATE 
Determine $\{M_i(t)\}_{i=1}^p$ such that $m_i(t) \leq M_i(t)$ $(i = 1, \ldots, p)$. 

\STATE 
Draw $\widetilde{\tau}_1, \ldots, \widetilde{\tau}_d$ such that $\PP( \widetilde{\tau}_i \geq t) = \exp \{ \int_0^t M_i(s) \, \dd s \}$.

\STATE 
Let $i_0 = \argmin_{i \in \{1, \dots, p\}} \{ \widetilde{\tau}_i \}$, and set $\tau =  \widetilde{\tau}_{i_{0}}$.

\STATE 
Evolve position:  $\{t^{(k+1)}, \pos^{(k+1)}\} = \{t^{(k)} + \tau, \pos^{(k)} + \theta^{(k)}  \tau \}$. 

\STATE 
Draw $u \sim \U\{(0,1)\}$. 

\STATE
{\bf if} $u < m_{i_0}(\tau)/M_{i_0}(\tau)$ {\bf then} $\vel^{(k+1)} = F_{i_{0}}\{ \vel^{(k)}\}$ {\bf else} $\vel^{(k+1)} = \vel^{(k)}$.

\ENDFOR
\end{algorithmic}
\textbf{Output:} A path of a zig-zag process specified by skeleton points $\{(\pos^{(k)}, \vel^{(k)})\}_{k=0}^{\Natt}$ and bouncing times $\{t^{(k)}\}_{k=0}^{\Natt}$. 
\end{algorithm}

\subsection{Use as a Monte Carlo method} 
\label{sec:zz_MCMC}

Ergodicity of the zig-zag process in the sense of $\EE_{\eta \sim \pi} \{\varphi(\eta)\} = \lim_{T \rightarrow \infty} T^{-1} \int_{0}^{T} \varphi\{\pos(t)\}\, \dd t$ holding almost surely for any $\pi$-integrable function $\varphi$ does not guarantee by itself reliable sampling due to a missing characterization of the statistical properties of the residual Monte Carlo error. The study of ergodic properties of piece-wise deterministic Markov processes has been a very active field of research in recent years. In particular, under certain technical conditions on the potential function $U$, exponential convergence in the law of the zig-zag process has been shown \citep{bierkens2017limit,andrieu2021hypocoercivity,bierkens2019ergodicity}. By \cite{Bhattacharya1982}, this implies a central limit theorem for the zig-zag process for a wide range of observables $\varphi : \RR^p \rightarrow \RR$, which justifies using it for sampling purposes: there is a constant ${\sigma}^{2}_{\varphi}>0$ such that $T^{1/2} \left [ \overline{\varphi}_{T}  - \EE_{\eta \sim \pi} \left \{ \varphi(\eta) \right \} \right ]$ with $\overline{\varphi}_{T}:=T^{-1} \int_{0}^{T} \varphi\{\pos(t)\}\, \dd t$ 
converges in distribution to $\N(0,{\sigma}_{\varphi}^{2})$ as $T \rightarrow \infty$. The constant ${\sigma}_{\varphi}^{2}>0$ is known as the asymptotic variance of the observable $\varphi$ under the process. In practice, the Monte Carlo estimate $\overline{\varphi}_{T}$ may either be computed by numerically integrating the observable $\varphi$ along the piece-wise linear trajectory of the zig-zag process, or by discretizing the trajectory, for example, by using a fixed step-size $\Delta t = T/N>0$ for some positive integer $N$, which results in a Monte Carlo estimate of the form $\widehat{\varphi}_{N} = (N+1)^{-1} \sum_{k=0}^{N} \varphi\{\xi(k \Delta t)\}$.

\section{Zig-zag sampler with generalized sub-sampling} \label{sec:lemma:proof}

We provide a complete algorithm for the zig-zag sampler with generalized sub-sampling as detailed in the main text and show that the target measure is an invariant measure of this process.

 \begin{proposition} \label{prop.zig-zag.generalized}
Let $\mu_{i}$ $(i = 1, \dots, p)$ be probability measures such that
$\widehat{ \partial}_{i} U^{\llh}(\pos,a)$, $a \sim \mu_{i}$, is an unbiased estimator of the $i$-th partial derivative of negative log-likelihood function $\partial_i U^{\llh}(\pos)$. Let 
$m^{\llh}_{i}(t, a) = \{ \vel_{i} \widehat{ \partial}_{i} U^{\llh}(\pos+t\vel,a) \}^+ $ and $M^{\llh}_{i}(t)$ be such that for all $a$ and $t \geq 0$, 
$m^{\llh}_{i}(t, a)  \leq M^{\llh}_{i}(t)$.  
Similarly, let $ m_{i}^0(t) =\{\vel_{i}\partial_{i}\Up(\pos+t\vel) \}^{+}$  and $M^0_{i}(t)$ be such that for all $t \geq 0$, $m_{i}^0(t) \leq M^0_{i}(t)$. 
Then the zig-zag process generated by \cref{alg.zig-zag.generalized}  preserves the target measure $\widetilde{\pi}$, and the effective bouncing rate of the generated zig-zag process is of the form $\lambda_{i}(\pos,\vel)= \lambda^0_{i}(\pos,\vel) + \lambda^{\llh}_{i}(\pos,\vel)$, with 
\begin{equation} \label{eq:def:rate:l}
\lambda_i^0 (\pos, \theta) 
= 
\left \{ \theta_i \partial_i  \Up(\pos) \right \}^{+}
\quad \text{and} \quad 
\lambda_i^{\llh}(\pos, \theta)  
=  
\EE_{a \sim \mu_{i}} \left [ \left \{ \vel_{i} \widehat{ \partial}_{i} U^{\llh}(\pos,a)  \right \}^{+} \right ].
\end{equation}

\end{proposition}

\begin{algorithm}
\caption{Zig-zag algorithm with generalized sub-sampling.} 
\label{alg.zig-zag.generalized} 

\textbf{Input:} Starting point $\pos^{(0)} \in \RR^p$, initial velocity $\theta^{(0)} \in \{-1,1 \}^p$, maximum number of bouncing attempts $\Natt$.

\begin{algorithmic}[1] 

\STATE 
Set $t^{(0)} = 0$.

\FOR{$k = 1, \ldots, \Natt$}

\STATE 
Draw $\widetilde{\tau}_i, \widehat{\tau}_i ~ (i =1,\ldots, p)$ such that $\PP( \widetilde{\tau}_i \geq t) = \exp \{ - \int_0^t M^{\llh}_i(s) \, \dd s \}$ and
$\PP( \widehat{\tau}_{i} \geq t) = \exp \{ - \int_0^t M^0_i(s) \, \dd s \}$.

\STATE 
Let $i_0 = \argmin_{i} \{ \widetilde{\tau}_i \}, ~j_{0}= \argmin_{i } \{ \widehat{\tau}_i \}$, and set $\tau =  \min \{\widetilde{\tau}_{i_{0}}, \widehat{\tau}_{j_{0}} \}$.

\STATE 
Evolve position: set $\{t^{(k+1)}, \pos^{(k+1)}\} = \{t^{(k)} + \tau, \pos^{(k)} + \theta^{(k)}  \tau \}$. 
\IF { $  \widehat{\tau}_{j_{0}}< \widetilde{\tau}_{i_{0}}$ }

\STATE
$\vel^{(k+1)} = F_{j_{0}}\{ \vel^{(k)}\}$. 

\ELSE 

\STATE 
Draw $B \sim \mu_{i}$ and $u \sim \U\{(0,1)\}$. 

\STATE 
{\bf if} $u < m^{\llh}_{i_0}(\tau,B)/M^{\llh}_{i_0}(\tau)$ {\bf then} $\vel^{(k+1)} = F_{i_{0}}\{ \vel^{(k)}\}$ {\bf else} $\vel^{(k+1)} = \vel^{(k)}$.

\ENDIF
\ENDFOR
\end{algorithmic}
\textbf{Output:} The path of a zig-zag process specified by skeleton points $\{(\pos^{(k)}, \vel^{(k)})\}_{k=0}^{\Natt}$ and bouncing times $\{t^{(k)}\}_{k=0}^{\Natt}$. 

\end{algorithm}
It follows straightforwardly from \cref{alg.zig-zag.generalized} that the refreshment rate $\gamma_{i}(\pos)$ induced by sub-sampling is 
\begin{equation} \label{eq:refresh}
\gamma_{i}(\pos) 
= 
\frac{ \EE_{a \sim \mu_{i}} \abs{\widehat{ \partial}_{i} U^{\llh}(\pos,a)} - \abs{\partial_{i} U^{\llh}(\pos)}}{2}.
\end{equation}
We have decomposed the stochastic rate function as the sum of a deterministic part $m_i^0(t)$ for the prior and a stochastic part $m_i^{\llh}(t, a)$ for the data, where the stochasticity is with respect to sub-sampling. This construction is related to ideas in \cite{bouchard2018bouncy}. The proof of \cref{prop.zig-zag.generalized}  generalizes that of Theorem 4.1 in \cite{bierkens2019zig}, and results similar to \cref{prop.zig-zag.generalized} can also be found in \cite{vanetti2017piecewise}.  

\begin{proof}[Proof of \cref{prop.zig-zag.generalized}]

The conditional probability of a bounce at time $t+\tau$ in the $i_{0}$-th coordinate given that $\widetilde{\tau}_{i_{0}} < \widehat{\tau}_{j_{0}}$ is 
\begin{equation*}
\EE_{a\sim\mu_{i_{0}}}  \left \{ \frac{m^{\llh}_{i_{0}}(\tau, a)  }{ M^{\llh}_{i_{0}}(\tau)} \right \} 
= 
\frac{\EE_{a\sim\mu_{i_{0}}}  \{ \vel_{i} \widehat{\partial}_{i} U^{\llh}(\pos+t\vel,a) \}^+  }{ M^{\llh}_{i_{0}}(\tau)} 
= 
M^{\llh}_{i_{0}}(\tau)^{-1} \lambda^{\llh}_{i_{0}}(\pos+t \vel,\vel),
\end{equation*}
with $\lambda_{i_{0}}^{\llh}$ as defined in \cref{eq:def:rate:l}.
Thus, since for all $t \geq 0,~\lambda^{\llh}_{i}(\pos+t,\vel) \leq M^{\llh}_{i}(t)$, it follows that the resulting process is indeed a thinned Poisson process with rate function $t\mapsto\lambda^{\llh}_{i}(\pos+t\vel,\vel)$. This rate function $\lambda^{\llh}_{i}$ satisfies 
\begin{align*}
\lambda^{\llh}_{i}(\pos,\vel) - \lambda^{\llh}_{i}\{\pos,F_{i}(\vel)\} 
& =  
\EE_{a\sim\mu_{i}}\left \{ \vel_{i} \widehat{\partial}_{i} U^{\llh}(\pos,a)  \right \}^+ 
- 
\EE_{a\sim\mu_{i}}\left \{- \vel_{i} \widehat{\partial_{i}} U^{\llh}(\pos,a)  \right \}^+
\\
& = 
\EE_{a \sim \mu_{i}}\left \{ \vel_{i} \widehat{\partial}_{i} U^{\llh}(\pos,a)  \right \}^+ -  \EE_{a\sim\mu_{i}}\left \{ \vel_{i} \widehat{\partial}_{i} U^{\llh}(\pos,a)  \right \}^-
\\
& =  
\EE_{a\sim\mu_{i}}  \left \{ \widehat{\partial}_{i} U^{\llh}(\pos,a) \right \} = \partial_{i} U^{\llh}(\pos). 
\end{align*}
Since also $\lambda^0_{i}(\pos,\vel) - \lambda^0_{i}\{\pos,F_{i}(\vel)\} =  \{ \vel_{i}  \partial_{i} \Up(\pos)  \}^+ -  \{ -\vel_{i}  \partial_{i} \Up(\pos)  \}^+ =  \partial_{i} \Up(\pos)$, it follows that the total rate $\lambda_{i}(\pos,\vel) = \lambda_{i}^0(\pos,\vel) + \lambda_{i}^{\llh}(\pos,\vel) $ satisfies $\lambda_i(\pos, \theta) - \lambda_i\{\pos, F_i(\theta)\} = \theta_i \partial_i U(\pos)$, $(i = 1, \dots, p)$. This is a sufficient condition for the zig-zag process to preserve $\widetilde{\pi}$ as its invariant measure \cite[Theorem 2.2]{bierkens2019zig}.

\end{proof}

\section{Importance sub-sampling using control variates} \label{sec:important}

As mentioned in the main text, control variates can be used to reduce the variance of the estimates of partial derivatives in the stochastic rate functions. When data points are sub-sampled uniformly, this is achieved as follows. Let  $\pos^\star$ denote a reference point, usually chosen as a mode of the posterior distribution. Unbiased estimates of the partial derivatives of $\partial_i U^{\llh}(\pos)$ can be obtained as
\begin{equation*} 
\widehat{\partial_{i}} U^{\llh}(\pos,J) 
=
\ndata \partial_i U^{J}(\pos  )   - 
\{  \ndata \partial_i U^{J}(\pos^{\star})   -  \partial_i U^{\llh}(\pos^\star)\},
\end{equation*}
where $J \sim \U [\{1,\dots, n\}]$ indexes the sampled data point, and the second term on the right-hand side is a control variate. If the partial derivatives $\partial_i U^{j}$ $(j=1,\dots,n; ~ i=1,\dots,p)$, are globally Lipschitz as specified in the following  \cref{as:lipschitz}, the corresponding stochastic rate functions $\widehat{m}_{i}^{J}(t) = \{ \vel_{i} \widehat{ \partial}_{i}U (\pos+t\theta,J) \}^{+}$ can be bounded by linear function of $t$.
\begin{assumption}\label{as:lipschitz}
The partial derivatives $\partial_i U^{j}$ are globally Lipschitz, that is, for suitable $r\geq 1$, there exist constants $C_{i}^{j}>0 ~ (i=1, \dots, p; j =1,\dots, n)$, such that for all $\pos_{1}, \pos_{2}\in \RR^p$, $| \partial_i U^{j}(\pos_{1} ) - \partial_i U^{j}(\pos_{2} ) |\leq C_{i}^{j} \norm{\pos_{1} - \pos_{2}}_{r}.$
\end{assumption}

More precisely, if  \cref{as:lipschitz} holds, realizations of $m_i^{\llh}(t,  J)$ can be bounded by $M^{\llh}_{i}(t) = \{ \vel_i \partial_i U^{\llh}(\pos^{\star})\}^{+} + \ndata C_{i} ( \norm{\pos - \pos^{\star}}_{r} + t p^{1/r} )$ for $C_{i} = n \max_{j \in \{1, \dots, n\}} C_{i}^{j}$.  In the case of the logistic regression problem, it is shown in \cite{bierkens2019zig} that \cref{as:lipschitz} holds for $r=2$ with 
\begin{align*}
C_i^j = \frac14 |x_i^j| \|x^j\|_2, ~~ (j=1,\dots,n; ~ i = 1, \dots, d).
\end{align*}

Analogously to how we generalized uniform sub-sampling to importance sub-sampling when uniform bounds are used (Section 3.2 of the main text), the control variate approach can be generalized to incorporate importance sub-sampling as follows. Consider the index $J$ to be sampled from a non-uniform probability distribution $\nu_{i}$, defined by $\nu_{i}[\{j\}] = \iweight_{i}^{j}$ $(j=1,\dots,n)$, where $\iweight_i^1, \dots, \iweight_i^n > 0$ are weights satisfying $\sum_{j=1}^n \iweight_i^j=1$. It  follows that
$\widehat{\partial}_i U^{\llh} (\pos,J) = \partial_i U^{\llh}(\pos^\star)  +  ( \iweight_{i}^{J} )^{-1} \{\partial_i U^{J}(\pos) - \partial_i U^{J}(\pos^\star) \}$, $J \sim \nu_{i}$ defines an unbiased estimator of $\partial_{i}U^{\llh}$, in which case the corresponding upper bounds are 
\begin{align*} 
M^{\llh}_{i}(t) 
= 
\left \{ \vel_i \, \partial_i \, U^{\llh}(\pos^{\star}) \right \}^{+} + \ndata \widetilde{C}_{i}(\iweight)\left ( \norm{\pos - \pos^{\star}}_{r} + t p^{1/r} \right )
~~ \text{with} ~~
\widetilde{C}_{i}(\iweight) 
= 
\max_{j \in \{1, \dots, n\}} \frac{C_{i}^{j}}{\iweight_{i}^{j}}.
\end{align*}
These upper bounds can be minimized by choosing $\iweight_{i}$ that minimizes $\widetilde{C}_{i}$, and this can be verified to be the case when $\iweight_{i}^{j} = C_{i}^{j}/ \overline{C}_{i}$ $(j=1,\dots, n)$ with $\overline{C}_{i} = \sum_{j=1}^n C_{i}^{j}$. 
Thus, optimal importance sub-sampling reduces the upper bounds from $(n/4) \max_{j \in \{1, \dots, n\}} |x_i^j| \| x^j \|_2$ to $(1/4) \sum_{j=1}^n |x_i^j| \| x^j \|_2$. Similar to the case when control variates are not used, if the covariates $x^j$s are sparse or possess outliers, this can result in large computational gains.

We point out that for efficient implementation of importance sub-sampling in combination with control variates, the decomposition of the bouncing rate into a likelihood rate and prior rate as in Algorithm 1 of the main text is critical, since in the presence of a prior, an estimate of the form
\[
\widehat{\partial}_{i} U (\pos,J) = \{ \ndata \partial_i U^{J}(\pos  )  + \partial_i \Up(\pos) \} - 
\{  \ndata \partial_i U^{J}(\pos^{\star}) + \partial_i \Up(\pos^{*})  -  \partial_i U(\pos^\star) \},
\]
with $J \sim \U [\{1,\dots, n\}]$, as for example used in \cite{bierkens2019zig}, results in $\pos$-dependent optimal importance weights. 

\section{Numerical examples for control variates}
\label{sec:numerical_CV}

\subsection{Scaling of computational efficiency}

We provide complementary results for Section 4.1 of the main text in the case when control variates are used. The relative gain in efficiencies is plotted in \cref{fig.scaling_bounds-CV}. Comparing with Figure 1 of the main text, we see that the gains are higher in the case when control variates are used than when they are not used. 
\begin{figure}[ht]
\centering 
\includegraphics[width=0.8\textwidth]{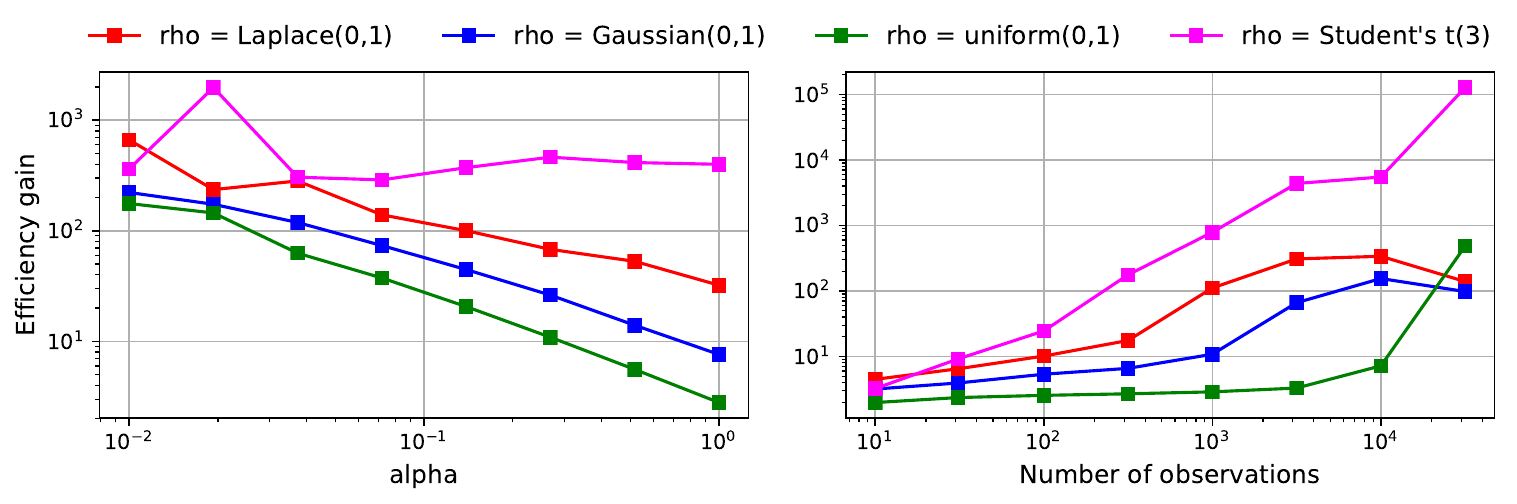}
\caption{Scaling of relative gain in efficiency when control variates are used.}
\label{fig.scaling_bounds-CV}
\end{figure}

\subsection{Dense data} 

We consider an example where the covariates are dense. We choose $n = 500$ and simulate the data as in Section 4.1 of the main text with $\alpha=0$. The prior is chosen as $\N_p(\mathbf{0},I_p)$. We choose $p \in \{10,50\}$ and compare importance sub-sampling with control variates and importance sub-sampling without control variates. Using control variates results in improved mixing of the zig-zag sampler for uniform sub-sampling when $p \ll n$, and the same is observed for importance sub-sampling in the left panel of \cref{fig.acfs_dense}. This is as expected since importance sub-sampling only reduces the upper computational bounds (and thus the total simulated time of the process), and not the diffusive properties of the resulting process. When $p$ is large relative to $n$, the posterior is not concentrated around a reference point, and this causes using control variates to be inefficient as compared to not using control variates, as seen in the right panel of \cref{fig.acfs_dense}. 

\begin{figure}[ht]
\centering 
\includegraphics[width=0.48\textwidth]{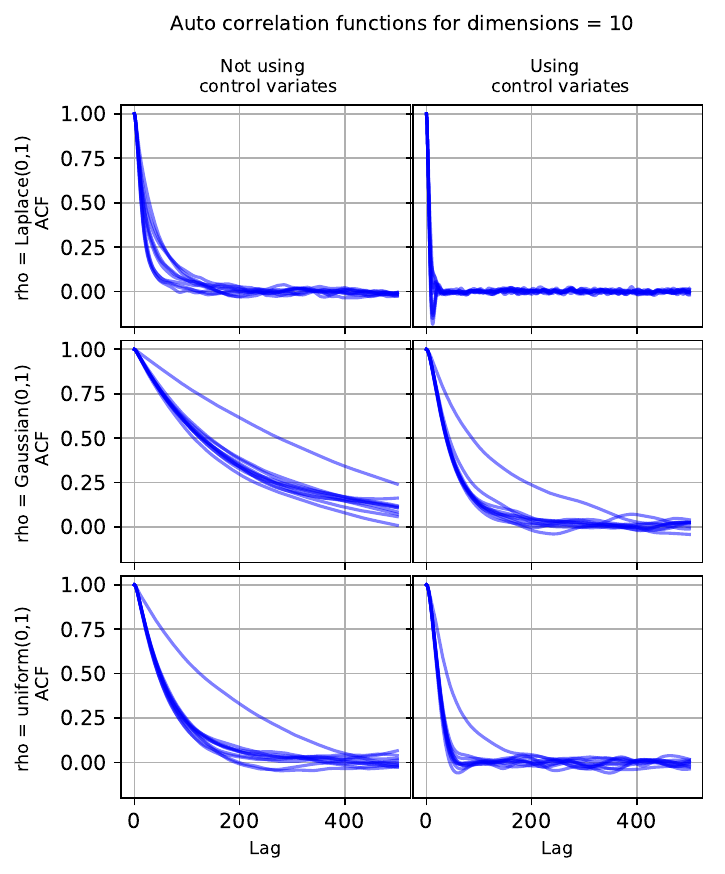}
\includegraphics[width=0.48\textwidth]{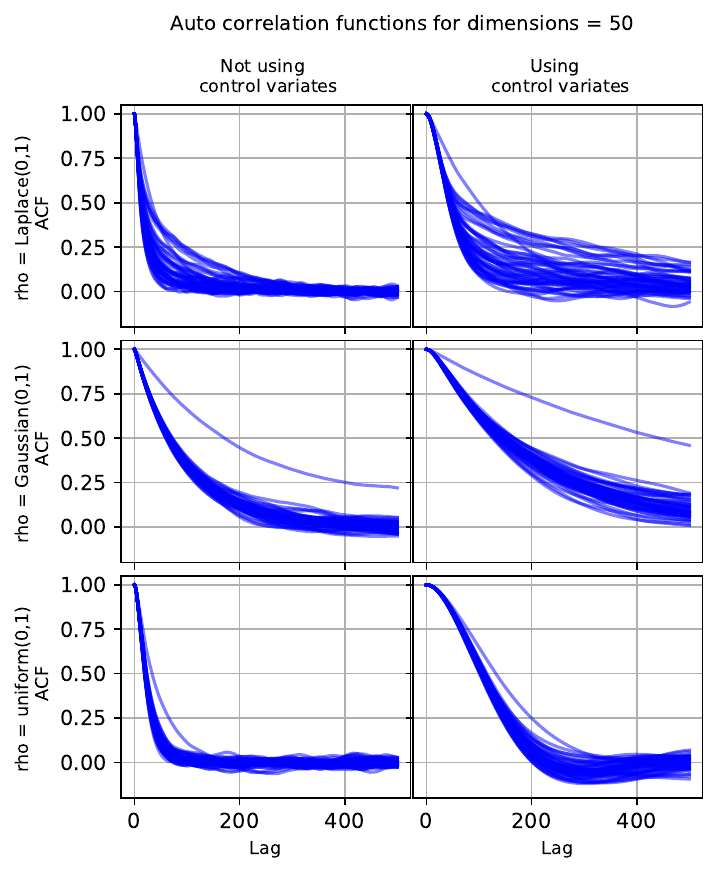}
\caption{Comparison of importance sub-sampling with and without control variates for dense covariates.}
\label{fig.acfs_dense}
\end{figure}

\subsection{Sparse data}

In Section 4.2 of the main text, we observed that control variates fail to perform efficiently as the responses become increasingly imbalanced. We plot the posterior variances of the parameters in the left panel of \cref{fig.imbalanced_responses_and_sparse_covariates} and observe that the posterior variances become increasingly large as the responses become increasingly imbalanced ($k$ decreases). This explains why control variates fail to perform efficiently in such a scenario. 

We conduct a similar experiment for sparse covariates. We use synthetic data and generate covariates as described in Section 4.1 of the main text for $\rho = \mathrm{Laplace}(0,1)$ and increasing level of sparsity among the covariates (that is, decreasing $\alpha$). We choose $p=5$ and $n=10^4$, and generate the responses such that half of them are ones (that is, the responses are perfectly balanced). The prior is chosen to be $\N_p(\mathbf{0}, I_p)$. We plot the ratio of the mixing time of the slowest mixing component for importance sub-sampling with control variates divided by the same for importance sub-sampling without control variates as a function of $\alpha$. As the covariates become sparser ($\alpha$ decreases), the performance of using control variates decreases relative to not using control variates, as seen in the middle panel of \cref{fig.imbalanced_responses_and_sparse_covariates}. We also plot the posterior variances in the right panel of \cref{fig.imbalanced_responses_and_sparse_covariates}, and observe that the posterior variances become increasingly large as the covariates become increasingly sparse.

\begin{figure}[ht]
\centering \includegraphics[width=\textwidth]{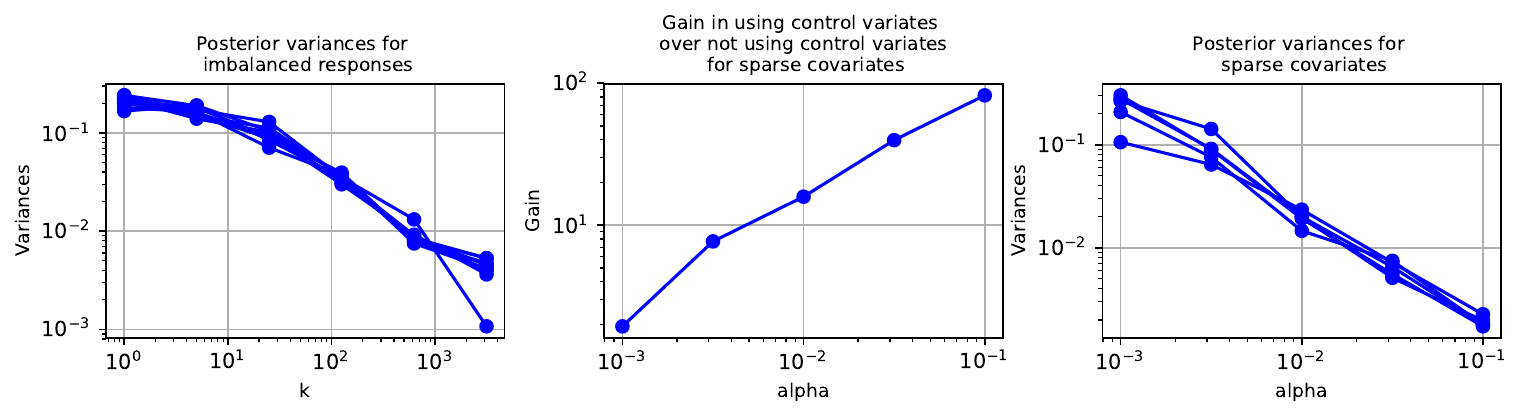}
\caption{Posterior variances for imbalanced responses (left panel), results for sparse covariates (center and right panels).}
\label{fig.imbalanced_responses_and_sparse_covariates}
\end{figure}

Finally, we consider scaling of the sub-sampling schemes with the number of observations $n$ when both responses are increasingly imbalanced and covariates are increasingly sparse. The covariates are again generated as described in Section 4.1 of the main text. We fix the dimension $p=200$ and choose the number of observations $n = 10^3, 5\times10^3, 10^4, 2\times10^4, 5\times10^4$, respectively. The corresponding data sets are simulated such that $\alpha = 10^{-1}, 2 \times 10^{-2}, 10^{-2}, 5\times10^{-3}$ and $2\times10^{-3}$, respectively, and $\rho = \N(0,1)$. 
The responses are simulated such that the proportion of ones are $10^{-2}, 2\times10^{-3}, 10^{-3}, 5\times10^{-4}$ and $2\times10^{-4}$, respectively. In this case, both the covariates become increasingly sparse and the responses become increasingly imbalanced as $n$ increases. We consider the mixing time of the slowest mixing component and plot this in \cref{fig.IACTs_scaling}, and observe that while importance sub-sampling with control variates improves on uniform sub-sampling with control variates, control variates are --as expected-- not efficient. This further supports the results seen in Section 4.2 of the main text. Importance sub-sampling without control variates is the most efficient sub-sampling strategy in this situation.

\begin{figure}[ht]
\centering \includegraphics[width=0.75\textwidth]{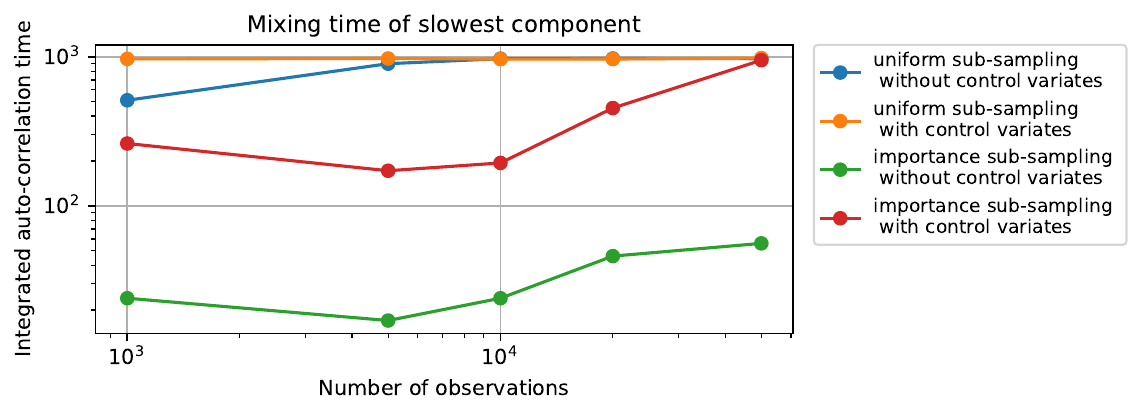}
\caption{Mixing of zig-zag sampler with different sub-sampling schemes.}
\label{fig.IACTs_scaling}
\end{figure}

\section{Stratified sub-sampling} \label{sec:strat:ss}

\subsection{General strategy}

As an alternative to control variate techniques, stratified sub-sampling can be used to reduce the variance of estimates of partial derivatives $\widehat{\partial}_i U$. Such approaches have been explored in the optimization literature \citep{zhao2014accelerating} and in the context of approximate Markov chain Monte Carlo \citep{fu2017cpsg}. The stratification techniques used in these references use strata constructed to minimize upper bounds for the variances of the noise in the estimates of the gradient, and are as such heuristically motivated and lack theoretical guarantees. In what follows, we describe a stratification scheme that is tailored to the particular setup of the zig-zag process and is guaranteed to reduce refreshment rates in the vicinity of a reference point $\pos^{\star}$.

Suppose we divide the data index set $\{1, \dots, n\}$ into $m$ strata $\{\strata_i^{k}\}_{k=1}^m$, which are such that for every component index $i \in \{1, \dots, p\}$, the sets $\{\strata_i^{k}\}_{k=1}^m$ form a partition of $\{1, \dots, n\}$; that is,
$\bigcup_{k=1}^m \strata_i^k = \{1, \dots, n\}$, and $\strata_i^k \cap \strata_i^{l} = \emptyset$ for $k \neq l$.  
If we consider a mini-batch $B$ to be sampled such that the $k$-th entry of $B$ is sampled independently and uniformly from the strata $\strata_{i}^{k}$, that is, $B = (J_{1},\dots,J_m)$, $J_{k} \sim \U (\strata_{i}^{k})$ $(k = 1, \dots, m)$,
then it is easy to verify that
\begin{align} \label{eq:def:U:strata}
\widehat{\partial}_i U^{\llh} (\pos,a) 
& =   
\sum_{k=1}^m \, \abs{\strata_{i}^{k}} \, \partial_i U^{J_k}(\pos), ~~ J_{k} \sim \U (\strata_{i}^{k}) ~~ (k = 1, \dots, m),
\end{align}
is an unbiased estimator for $\partial_{i} U^{\llh}(\pos)$. When constructing upper bounds for the corresponding stochastic rate function, stratified sub-sampling can only improve upon the upper bounds used in uniform sub-sampling, since
\begin{align*} 
m^{\llh}_{i}(t, B) 
& = 
\left \{ \vel_i \sum_{k=1}^m \abs{\strata_{i}^{k}} \, \partial_i U^{J_k}(\pos+t\vel) \right \}^{+}  
\leq 
\sum_{k=1}^m \abs{\strata_{i}^{k}} \left \{ \vel_i \partial_i U^{J_k}(\pos+t\vel) \right \}^{+}  
\\
& \leq 
\sum_{k=1}^m \abs{\strata_{i}^{k}} \frac{M^{\llh}_{i}(t)}{n}
= 
M^{\llh}_{i}(t),
\end{align*} 
where $\{M^{\llh}_i(t)\}_{i=1}^p$ are the upper bounds given in Section 3.1 of the main text. 
This thus leads to a zig-zag process with effective bouncing rates $\lambda^{\llh}_{i}(\pos,\vel) = |\widetilde{\strata}|^{-1} \sum_{(j_{1},\dots,j_m) \in \widetilde{\strata} } \{  \vel_{i}\sum_{k=1}^m \abs{\strata_{i}^{k}}\, \partial_i U^{j_k} (\pos) \}^{+}$,
where $\widetilde{\strata} = \strata_{i}^{1}  \times \dots \times \strata_{i}^m$.

This stratification scheme differs from other related approaches in that we consider separate stratification in each dimension, which allows for a more effective variance reduction as strata on the index set can be tailored to the particular covariance structure of the partial derivatives in each dimension. For iterative methods requiring a synchronous computation of partial derivatives in all dimensions, such a stratification may result in additional computation overhead, whereas in the case of the zig-zag sampler, partial derivatives are computed asynchronously anyway, so we can incorporate such a stratification with no additional computational overhead. 

\subsection{Construction of strata}

In what follows, we describe how gradient information at a reference point $\pos^{\star} \in \RR^p$ can be used to construct strata for the estimator described above so that under the regularity conditions of \cref{as:lipschitz}, the refreshment rate of the associated zig-zag process is provably reduced within some vicinity of $\pos^{\star}$ in comparison to the refreshment rate induced by uniform sub-sampling using the same mini-batch size.

By \cref{eq:refresh} it follows that for given strata $\strata_{i}^{k} ~ (k=1,\dots,m)$, the effective refreshment rate $\gamma_{i}[ \pos, \{\strata_{i}^{k}\}_{k=1}^{m}]$ which is induced by stratified sub-sampling  is of the form
\begin{equation*}
\gamma_{i}[ \pos, \{\strata_{i}^{k}\}_{k=1}^{m} ]
= 
\frac{1}{2}\left \{ \EE_{B \sim \mu_{i}} \abs*{\sum_{k=1}^m \abs{\strata_{i}^{k}}\, \partial_i U^{J_k} (\pos)} - \abs{\partial_{i} U^{\llh}(\pos)} \right \}.
\end{equation*}
If \cref{as:lipschitz} is satisfied, it can be easily derived that
\begin{equation}\label{eq:strat:bound:1}
\gamma_{i}[ \pos, \{\strata_{i}^{k}\}_{k=1}^{m} ] 
\leq
\gamma_{i}[ \pos^\star, \{\strata_{i}^{k}\}_{k=1}^{m} ]  + n C_i \norm{\pos - \pos^{\star}}_{r},
\end{equation}
for all $\pos \in \RR^p$, since 
\begin{align*}
& 
\max_{(j_{1},\dots,j_{k}) \in \widetilde{\strata} } \left | \sum_{k=1}^m \abs{\strata_{i}^{k}}\, \partial_{i}U^{j_{k}}(\pos) - \sum_{k=1}^m \abs{\strata_{i}^{k}}\, \partial_{i}U^{j_{k}}(\pos^{\star}) \right | \\
& \leq 
\ndata \max_{j \in \{1, \dots, \ndata\} } |\partial_{i}U^{j}(\pos) - \partial_{i}U^{j}(\pos^{\star})| \leq n  C_{i}\norm{\pos - \pos^{\star}}_{r}
\end{align*}
for all $\pos \in \RR^p$,
and similarly 
\begin{equation}\label{eq:strat:bound:2}
\gamma_{i}[ \pos, \{\strata_{i}^{k}\}_{k=1}^{m} ] 
\geq
\gamma_{i}[ \pos^\star, \{\strata_{i}^{k}\}_{k=1}^{m} ]   - n C_i \times \norm{\pos - \pos^{\star}}_{r}
\end{equation}
for all $\pos \in \RR^p$.
The inequality \eqref{eq:strat:bound:1} suggests that in order to minimize the refreshment rate $\gamma_{i}$ in the vicinity of the  reference point $\pos^{\star}$, the strata should be constructed in each dimension $i=1,\dots,p$, such that $\gamma_{i}(\pos^{\star})$ is minimized, that is, we construct the strata as the solution to the minimization problem 
\begin{equation}\label{eq:mini}
\{\strata_{i}^{k}\}_{k=1}^m 
= 
\argmin_{\{\widetilde{\strata}_{i}^{k}\}_{k=1}^m} \, \EE_{B \sim \mu_{i}} \abs*{\sum_{k=1}^m \abs{\widetilde{\strata}_{i}^{k}}\, 
\partial_i U^{j_k} (\pos^{\star})}.
\end{equation}
Since solving \cref{eq:mini} may be hard if the number of strata $m$ is large, we consider strata constructed as the solution to the minimization problem
\begin{equation}\label{eq:mini:2}
\{\strata_{i}^{k}\}_{k=1}^m 
= 
\argmin_{\{\widetilde{\strata}_{i}^{k}\}_{k=1}^m}  \sum_{k=1}^m \abs{\widetilde{\strata}_{i}^{k}}\mathrm{diam} \left [ \left \{ \partial_i U^{j} (\pos^{\star}): j\in \widetilde{\strata}_{i}^{k} \right \}\right ],
\end{equation}
where $\pos^{\star}$ denotes the reference point and $\mathrm{diam}(A)= \max (A) - \min(A)$ denotes the diameter of a finite set $A\subset \RR$. This approximation is justified by \cref{lem:inequality}, since the objective function in \eqref{eq:mini:2} is an upper bound of the objective function in \eqref{eq:mini} up to a constant. 

Let $\gamma_{i}^{\llh,  \, (m)}(\pos)$ be the refreshment rate induced by uniform sub-sampling using a mini-batch of size $m$. By applying the inequalities \eqref{eq:strat:bound:1} and \eqref{eq:strat:bound:2} to  $\gamma_{i}[ \pos,\{\strata_{i}^{k}\} ]$ and $\gamma_{i}^{\llh,  \, (m)}(\pos)$, respectively, it follows that the refreshment rate is indeed reduced for all points within the sphere centered at $\pos^{\star}$ and radius $(2nC_{i})^{-1}[ \gamma_{i}(\pos^{\star}) - \gamma_{i}[ \pos^{\star},\{\strata_{i}^{k}\}_{k=1}^{m}] ]$, that is, 
$\gamma_{i}[ \pos,\{\strata_{i}^{k}\}] \leq \gamma_{i}^{\llh,  \, (m)}(\pos)$ for all $\pos \in \RR^{p}$ satisfying
\begin{equation*}
\norm{\pos-\pos^{\star}}_{r} \leq  \frac{1}{2nC_{i}}\left [ \gamma_{i}(\pos^{\star}) - \gamma_{i}[ \pos^{\star},\{\strata_{i}^{k}\}_{k=1}^{m}] \right ].
\end{equation*}

\begin{remark}
The stratification scheme described relies on the same regularity assumptions in terms of $U_{i}^{j}$ $(j=1,\dots,n; ~ i=1,\dots,p)$ as the control variate approaches in order to ensure theoretical guarantees. Moreover, similarly as in the discussed control variate approaches, the method relies on the posterior distribution to be concentrated around a single mode in order for the stratification to be effective. However, unlike the described control variate approaches, stratified sub-sampling is used in combination with uniform bounds, which allows for efficient application of the approach when only overly conservative estimates of Lipschitz constants are available. 
\end{remark}

\subsection{Stratification algorithm} \label{sec:clustering_algorithms}

\cref{alg.clustering} provides an implementation of the clustering algorithm used to obtain strata in the case that gradient information at a reference point $\pos^{\star}$ is available. The algorithm computes 
for input $x_{j}= \partial_i U^{j}(\pos^{\star})$ $(j=1,\dots,\ndata)$, a partition of the set $\{x_{1},\dots,x_{\ndata}\}$. The corresponding partition of the index set $\{1,\dots,\ndata\}$ is an approximate solution of the minimization problem of equation (17) in the main text. Define functions $f_{\textit{score}}$ and $f_{\textit{split-score}}$ as follows for $(r = 1, \dots, n; \, k=1, \dots, r)$:
\begin{align*}
f_{\textit{score}}\{ (x_{1},\dots,x_{r})\} 
& = 
r \left (  \max_{k \in\{1,\dots,r\}} x_{k} - \min_{k \in\{1,\dots,r\}} x_{k}  \right),
\\
f_{\textit{split-score}} \{(x_{1},\dots,x_{r}),k\}
& = 
k f_{\textit{score}}\{(x_{1},\dots,x_{k})\} +  (r-k)  f_{\textit{score}}\{(x_{k+1},\dots,x_{r})\}.
\end{align*}
The algorithm is as follows. 

\begin{algorithm}[ht]
\caption{Greedy clustering } 
\label{alg.clustering} 

\textbf{Input:} Sorted data points  $x_{1} \leq \dots \leq x_{n}$, number of clusters $m$.

\begin{algorithmic}[1] 

\STATE
Set $G^{0} =  \{ ( x_{j} )_{j=1}^n \}$
\FOR{$l = 1, \ldots, m$}
\STATE
Set $(\tilde{x}_{1},\dots,\tilde{x}_{s}) = \argmin_{\x \in G^{l-1}} \min_{k}  \left \{ f_{\textit{split-score}} (\x,k) - f_{\textit{score}}(\x)  \right \}$.
\STATE
Set $k_{0} = \argmin_{k} f_{\textit{split-score}}  \{ (\tilde{x}_{1},\dots,\tilde{x}_{s}),k \}$.
\STATE
Set $G^{l} =  G^{l-1} \setminus \left \{ (\tilde{x}_{1},\dots,\tilde{x}_{s}) \right \}  \cup \left \{ (\tilde{x}_{1},\dots,\tilde{x}_{k_{0}}), \,  (\tilde{x}_{k_{0}+1},\dots,\tilde{x}_{s}) \right \}$.

\ENDFOR

\end{algorithmic}
\textbf{Output:} A set of vectors $\{ \x^{k} \}_{k=1}^m =G^{m}$ defining a partition of the set $\{x_{1}, \dots, x_{n}\}$.
\end{algorithm}

The computational cost of the proposed stratification algorithm is $\OO(n \log n)$ in each dimension, and this can be trivially parallelized in the number of dimensions $p$. 

\begin{remark}[A note on computational aspects]
Computing the strata is a one-off operation that costs $\OO(n \log n)$ in each dimension. The cost for uniform sub-sampling of a single data point is $\OO(1)$. For importance sub-sampling, the cost is $\OO(n)$ for a naive implementation, which reduces down to $\OO(m)$ when the covariates are sparse, where $m$ denotes the number of non-zero weights. These can be brought down
to $\OO(\log n)$ and  $\OO(\log m)$ for dense and sparse covariates, respectively, for example, using the techniques of \cite{wong1980efficient}.
The pre-constant for the order is small compared to the cost of all other operations. 
We wrote code in Julia \citep{bezanson2017julia} and used the \texttt{wsample} function for weighted sub-sampling, which is slower than uniform sub-sampling. 
Using the \texttt{wsample} function leads to the sub-sampling for stratified sub-sampling with importance weights being slightly slower than uniform sub-sampling as well.

\end{remark}

\section{Proofs} \label{sec:proofs}

\subsection{From main text}

\begin{proof}[of Lemma 1 of main text]
Let $B= (j_{1},\dots,j_{m})$ for $B \in \{1, \dots, n\}^{m}$. The result follows as
\begin{align*}
\lambda_{i}^{\llh, \, (m)}(\vel,\pos) 
& = 
\frac{1}{n^{m}}\frac{1}{m} \sum_{B \in \{1, \dots, n\}^{m}}\left \{ \sum_{k=1}^{m} \vel_{i} \partial_i U^{j_{k}} (\pos) \right \}^{+} 
= 
\frac{1}{n^{m+1}}\sum_{B \in \{1, \dots, n\}^{m+1}}\frac{1}{m} \left \{ \sum_{k=1}^{m} \vel_{i} \partial_i U^{j_{k}} (\pos) \right \}^{+} 
\\
& = 
\frac{1}{n^{m+1}}\frac{1}{m+1} \sum_{B \in \{1, \dots, n\}^{m+1}} \sum_{s=1}^{m+1}\frac{1}{m} \left \{ \sum_{k\in\{1, \dots, m+1\}\setminus\{s\}} \vel_{i} \partial_i U^{j_{k}} (\pos) \right \}^{+} 
\\
& \geq 
\frac{1}{n^{m+1}}\frac{1}{m+1} \sum_{B \in \{1, \dots, n\}^{m+1}} \left \{\sum_{s=1}^{m+1}\frac{1}{m}  \sum_{k\in\{1, \dots, m+1\}\setminus\{s\}} \vel_{i} \partial_i U^{j_{k}} (\pos) \right \}^{+} 
\\
& = 
\frac{1}{n^{m+1}}\frac{1}{m+1} \sum_{B \in \{1, \dots, n\}^{m+1}} \left \{\sum_{k\in\{1, \dots, m+1\}} \vel_{i} \partial_i U^{j_{k}} (\pos) \right \}^{+} 
= \lambda_{i}^{\llh, \, (m+1)}(\vel,\pos). \qedhere
\end{align*}
\end{proof}

\subsection{Additional Lemmata}

\begin{lemma}\label{lem:inequality}
The following inequality holds for any component index $i\in\{1,\dots,p\}$ and any partition $\{\strata_{i}^{k}\}_{k=1}^{m}$ of the index set $\{1,\dots,n\}$.
\begin{equation*}
 \abs*{\sum_{k=1}^m \abs{\strata_{i}^{k}}\, \partial_i U^{j_{k}}(\pos)} - \abs{\partial_{i} U^{\llh}(\pos)}  \leq  \sum_{k=1}^n \abs*{\strata_{i}^{k}}\mathrm{diam}\left [\{ \partial_i U^{j}(\pos): j\in \strata_{i}^{k}\}\right ], 
\end{equation*}
\end{lemma}
\begin{proof}
Let $\overline{E}_{i}^{k}(\pos) = \abs{\strata_{i}^{k}}^{-1} \sum_{j \in \strata_{i}^{k}} \partial_i U^{j}(\pos)$.
Then,
\begin{equation*}
\begin{aligned}
 \abs*{\sum_{k=1}^{m} \abs{\strata_{i}^{k}}\, \partial_i U^{j_{k}}(\pos)} - \abs{\partial_{i} U^{\llh}(\pos)} 
 &\leq 
 \abs*{ \sum_{k=1}^{m} \abs{\strata_{i}^{k}} \left \{  \partial_i U^{j_{k}}(\pos) -  \overline{E}_{i}^{k}(\pos)  \right \} }  + \abs*{\sum_{k=1}^{m} \abs{\strata_{i}^{k}}\overline{E}_{i}^{k}(\pos)} - \abs{\partial_{i} U^{\llh}(\pos)}\\
 &\leq 
 \sum_{k=1}^m \abs*{\strata_{i}^{k}}\mathrm{diam}\left [\{ \partial_i U^{j}(\pos): j\in \strata_{i}^{k}\}\right ],
\end{aligned}
\end{equation*}
where the first and second inequality follow as applications of the triangular inequality, and the fact that $\sum_{k=1}^{m} \abs{\strata_{i}^{k}} \, \overline{E}_{i}^{k}(\pos) = \partial_{i} U^{\llh}(\pos)$.
\end{proof}

\section{Pre-conditioning} \label{sec.precond}

\subsection{General recipe for pre-conditioning} 

We propose an adaptive pre-conditioning variant of the zig-zag process,  which can rapidly learn how to pre-condition using initial samples. To achieve this, we modify the velocity by giving it a speed in addition to a direction. The standard zig-zag process has unit speed in all dimensions; indeed, $\vel \in \{-1, 1\}^p$. We still use $\vel$ to denote the direction of the velocity, and in addition, we introduce $\alpha \in \RR_+^p$ to denote the speed. Thus, the overall velocity of the zig-zag process is now $\theta \alpha = (\theta_1 \alpha_1, \dots, \theta_p \alpha_p)$. As mentioned in \cite{bierkens2019zig}, such an extension of the zig-zag process indeed preserves the target measure $\pi$. 
The zig-zag process with adaptive pre-conditioning is constructed so that the velocity vector only changes at bouncing events. At each bouncing event, two things happen. First, the direction of the velocity is modified by flipping the sign of one component of $\theta$ as described in Section 2.2 of the main text. Second, the speed of the process is also modified as described below. 
The deterministic motion (equation (2) of the main text) is modified to $\pos(t) = \pos + \theta \alpha t$.

To construct computational bounds for this modified zig-zag process, observe that the rates from \cref{prop.zig-zag.generalized} are generalized as $m^{\llh}_{i}(t, \alpha, a) = \{ \vel_{i} \alpha_i \widehat{ \partial_{i}} U^{\llh}(\pos+t\vel \alpha, a) \}^+$ and $m_{i}^0(t, \alpha) = \{\vel_{i} \alpha_i \partial_{i}\Up(\pos + t \vel \alpha) \}^{+}$.
Under Assumption 1 of the main text (not using control variates), realizations of the stochastic rate function $m^{\llh}_{i}(t, \alpha, a)$ can be bounded by $\alpha_i \max_{j \in \{1, \dots, n\}} c_i^j$; we thus have a dynamically evolving bound. 
Similarly, under \cref{as:lipschitz} (using control variates), realizations of the stochastic rate function $m^{\llh}_{i}(t, \alpha, a)$ can be bounded by
$\left \{ \vel_i \alpha_i \partial_i U^{\llh}(\pos^{\star}) \right \}^{+} + \ndata \alpha_i C_{i} ( \|\pos - \pos^{\star}\|_{r} + t \, \|\alpha\|_r )$ for $C_{i} = 
\ndata \max_{j \in \{1, \dots, n\}} C_{i}^{j}$.

\subsection{Pre-conditioning scheme} \label{sec.precond_scheme}

We consider an adaptive pre-conditioning scheme using covariance information of the trajectory history. An intuitive reasoning behind this is as follows.
Consider a zig-zag process on $d=2$ dimensions. Suppose the standard deviation in dimension one is twice that in dimension two. 
Intuitively, the process should move twice as fast in the first dimension as in the second dimension for it to mix well. 
Since the standard deviations across dimensions are unknown, we estimate them on the run by storing the first moment and second moment of the process, $\widehat{\mu}^{[1]}(t) = t^{-1}  \int_0^t \pos(s) \, \dd s$ and  $\widehat{\mu}^{[2]}(t) = t^{-1} \int_0^t \pos(s)^2 \, \dd s$, and estimating the standard deviation at time $t$ as $\widehat{\mathrm{sd}}(t) = \{\widehat{\mu}^{[2]}(t) -  \widehat{\mu}^{[1]}(t)^2\}^{1/2}$.
We update the moments $\widehat{\mu}^{[1]}(t)$ and $\widehat{\mu}^{[2]}(t)$ at every bouncing event and choose the speed $\alpha$ as $\alpha_i = d \times \widehat{\mathrm{sd}}(t)_i / \{\sum_{i=1}^p \widehat{\mathrm{sd}}(t)_i\}$ $(i = 1, \dots, p)$, where $\widehat{\mathrm{sd}}(t) = (\widehat{\mathrm{sd}}(t)_1, \dots, \widehat{\mathrm{sd}}(t)_p)$. We have normalized the speed vector to make it sum to $p$ to remain consistent with the usual zig-zag process. 

\subsection{Optimal precondition matrix for Gaussian target with independent components} \label{sec.opt_precond}

Consider a Gaussian target with independent components specified by the potential function $U(\pos) =  \pos^{\trans}\Omega \pos/2$, where $\Omega = {\rm diag}(\omega_{1},\dots, \omega_{p})$ is a diagonal matrix. As shown in \cite{bierkens2019zig}, the generator $\Lc$ of the zig-zag process with speed $\alpha_{i}>0$ in the component $ i=1,\dots,p$ is of the form 
\begin{equation} \label{eq:gen:precond:zz}
\Lc \varphi(\pos,\vel) 
= 
\sum_{i=1}^p
\left [ \alpha_{i}\vel_{i}\partial_{i}\varphi(\pos,\vel) 
+ 
\lambda_{i}(\pos,\vel)  \varphi \{ \pos,F_{i}[\vel] \} \right ] 
-
\varphi(\pos,\vel),
\end{equation}
where $\varphi$ is a suitable observable defined on the domain of the generator $\Lc$. In what follows, we show that the asymptotic variance of the moments of the target is minimized if the speed coefficients are chosen as $\alpha_{i} = \omega_{i}^{-1/2}$.
More precisely, let $\varphi_{i,k}(\pos) = \pos_{i}^{k}$ be the $k$-th power of $\pos_i$ and let $\sigma_{\varphi_{i,k}}^{2}$ be the asymptotic variance of this observable for the zig-zag process. Under the constraint that the speed of all components sum up to $p$, that is, $(\alpha_{1},\dots,\alpha_{p}) \in p \, \Delta^{p-1}$, where $\Delta^{p-1}$ denotes the standard simplex in $\RR^p$, we show that for any moment index $k$, the speed vector which minimizes $\max_{i \in \{1,\dots, p\}} \sigma_{\varphi_{i,k}}^{2}$ is $\alpha_{i} = p \, \omega_{i}^{1/2}/(\sum_{i=1}^p\omega_{i}^{1/2})$ $(i=1, \dots, p)$, that is, 
\begin{equation} \label{eq:sol:opt:omega}
\left ( \frac{p \, \omega_{i}^{1/2}}{\sum_{i=1}^p\omega_{i}^{1/2}} \right )_{1\leq i \leq p}
= 
\argmin_{(\alpha_{1},\dots,\alpha_{p})\in p \, \Delta^{p-1}} \max_{i\in \{1,\dots,p\}}\sigma_{\varphi_{i,k}}^{2}
\end{equation}
To show this, we first note that in the case of a Gaussian target, the rate function $\lambda_{i}(\pos,\vel)$ in \cref{eq:gen:precond:zz} has the explicit form $\lambda_{i}(\pos,\vel) = \left ( \alpha_{i} \vel_{i} \omega_{i}\pos_{i} \right )^{+}$, and thus $\Lc = \sum_{i=1}^p\alpha_{i} \Lc_{i}$
with $\Lc_{i}\varphi(\pos,\vel) = \vel_{i}\partial_{i}\varphi(\pos,\vel) + ( \vel_{i}\omega_{i} \pos_{i} )^{+} [ \varphi\{\pos,F_{i}(\vel)\} - \varphi(\pos,\vel) ]$. 
From \citealp[Example 3.3]{bierkens2017limit}, it follows that the integrated auto-correlation time for each moment $\pos_{i}^{k}, k\in \mathbb{N}$, is proportional to $\omega_{i}^{1/2}$. Moreover, since in the considered setup the zig-zag process converges exponentially fast in law, we can write the asymptotic variance of $\varphi_{i,k}$ as $\sigma^{2}_{\varphi_{i,k}} = -2 \int \{( \alpha_{i}\Lc_{i})^{-1} \varphi_{i,k}(\pos) \} \varphi_{i,k}(\pos) \, \target(\dd \pos)$ \citep[Proposition 9]{Lelievre2016a}. Thus, it follows that $\sigma_{\varphi_{i,k}}^{2}$ is linear in $\alpha_{i}^{-1}$. Therefore, $\sigma_{\varphi_{i,k}}^{2} = c_{k} \alpha_{i}^{-1} \omega_{i}$ with $c_{k}= -2 \int \{\Lc_{i}^{-1} \varphi_{i,k}(\pos) \}\varphi_{i,k}(\pos) \, \target(\dd \pos)$, which implies \cref{eq:sol:opt:omega}.

\subsection{Numerical example}

Recalling the high-dimensional sparse example from Section 4.3 of the main text, we plot the auto-correlation function for the zig-zag process with importance sub-sampling and adaptive-preconditioning in \cref{fig.highdim_precond}. Comparing this to the right panel of Figure 2 of the main text, we observe that adaptive pre-conditioning can improve performance significantly.

\begin{figure}[H]
\centering 
\includegraphics[width=0.48\textwidth]{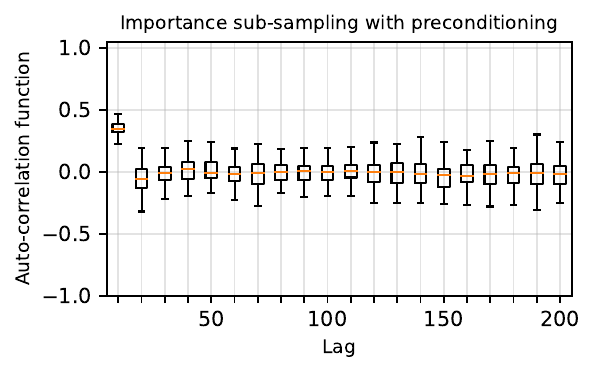}
\caption{Auto-correlation functions for high-dimensional sparse example with adaptive pre-conditioning.}
\label{fig.highdim_precond}
\end{figure}

\section{Other priors} \label{sec.different_priors}

\subsection{A weakly informative prior} 

\cite{gelman2008weakly} propose using independent Cauchy distributions with mean zero and scale 2.5 as the prior for the $\xi_i$'s in logistic regression. Then  
\begin{align*}
p_0(\pos) 
& = 
\prod_{i=1}^p \frac1{2.5 \left\{ 1 + (\pos_i/2.5)^2 \right \} } \\
\Rightarrow 
\Up(\pos) 
& = 
- \log p_0(\pos) 
= 
p \log 2.5 + \sum_{i=1}^p \log \left \{ 1 + \left ( \frac{\xi_i}{2.5} \right )^2 \right \} 
\\
\Rightarrow \partial_i \Up (\pos) 
& = 
\frac{2 \, \pos_i / 2.5 }{1 + \left ( \xi_i/2.5 \right )^2 } 
\Rightarrow 
\left | \partial_i U^0(\pos) \right |
\leq 
\frac{2}{2.5} | \pos_i |,  
\end{align*}
and, therefore, this is a linear bound as in the Gaussian prior case.

\subsection{Generalized double Pareto prior} 

A generalized double Pareto prior for Bayesian shrinkage estimation and inferences was proposed by \cite{armagan2013generalized}. The generalized double Pareto density is given by 
\begin{align*}
p(\pos_i \mid \theta, \alpha) 
& = 
\frac1{2 \theta} \left ( 1 + \frac{|\pos_i|}{\alpha \, \theta} \right )^{-(1+\alpha)},
\end{align*}
where $\theta>0$ is a scale parameter and $\alpha>0$ is a shape parameter. Thus we have 
\begin{align*}
\partial_i U^0(\pos) 
& =
(1+\alpha) \, \partial_i \log \left ( 1 + \frac{|\pos_i|}{\alpha \, \theta} \right ) 
= 
\text{sign}(\pos_i) \frac{1+\alpha}{\alpha \, \theta + |\pos_i|}, \quad \pos_i \neq 0 \\
\Rightarrow \left | \partial_i U^0(\pos) \right | 
& \leq 
\frac{1+\alpha}{\alpha \, \theta},
\end{align*}
which is a constant bound.

\subsection{Laplace prior} 

Another common prior is the double exponential prior, also known as the Laplace prior \citep{williams1995bayesian, park2008bayesian}. The prior is 
\begin{align*}
p(\pos) 
& = 
\prod_{i=1}^p \frac1{2b} \exp \left ( - \frac{|\pos_i|}b \right )
\end{align*}
for some $b \geq 0$. For this prior, $\partial_i U^0(\pos) = \partial_i |\pos_i|/b
\Rightarrow | \partial_i U^0(\pos) | \leq 1/b$, and this is also a constant bound.

\end{document}